\documentclass[aps,pra,reprint,twocolumn,nofootinbib,superscriptaddress,floatfix]{revtex4-2}

\usepackage[english]{babel}
\usepackage[utf8]{inputenc}
\usepackage[normalem]{ulem}
\usepackage{amsmath,amssymb,amsfonts,amsthm,physics,graphicx,xcolor,mathtools,dsfont,bm,hyperref}
\usepackage{quantikz}
\hypersetup{
	colorlinks=true,
	linkcolor={red!40!black},
	citecolor={blue!60!black},
	urlcolor={blue!50!black}
	}
\graphicspath{{fig/}}
\definecolor{myred}{rgb}{0.85,0,0}
\definecolor{mygray}{rgb}{0.87, 0.87, 0.87}

\let\oldbibitem\bibitem 
\renewcommand{\bibitem}{
    \renewcommand{\doi}[1]{\texttt{\href{https://doi.org/##1}{doi:##1}}} 
    \let\bibitem\oldbibitem 
    \oldbibitem 
}

\newcommand{\GG}{\mathcal{G}}
\newcommand{\VV}{\mathcal{V}}
\newcommand{\EE}{\mathcal{E}}

\newtheorem{lemma}{Lemma}
\newtheorem{proposition}{Proposition}

\begin{document}

\title{Learning structural balance of graphs from quantum spectral features}
\author{Stefano Scali}
\email{scali.stefano@gmail.com}
\affiliation{Department of Physics and Astronomy, University of Exeter, Exeter EX4 4QL, UK}
\author{Oleksandr Kyriienko}
\affiliation{School of Mathematical and Physical Sciences, University of Sheffield, Sheffield S10 2TN, UK}

\begin{abstract}
We develop a quantum approach to spectral feature extraction from the density of states (DOS) of a problem-dependent Hamiltonian, and apply it to machine learning on signed graphs. We propose to embed a signed graph as an Ising model instance with positive and negative interactions, and use the standardized moments of the Ising DOS as features for learning. We show that these moments count signed closed walks, are switching-invariant, and are size-free by construction. As a benchmark, we target learning the frustration index, an NP-hard measure of structural balance that can be labeled exactly at moderate size. At zero field, the models can be sampled classically, allowing the quantum extraction procedure to be certified against exact ground truth. We propose DOS-QPE, a phase estimation on a purified maximally mixed probe, which samples the spectral density with orders of magnitude fewer shots than Hadamard test-based trace sampling and feeds the resulting features directly into classically trained models. On $1.4\times10^5$ labeled graphs the exact DOS determines the frustration index, and five moments recover it with a mean error of $0.4$, well below one sign flip. Beyond zero field, the underlying trace-estimation problem is DQC1-complete, providing access to spectral features for which no efficient classical sampling method is known. Our work opens routes towards quantum applications in social network balance analysis, spin-glass studies, correlation clustering, and protein-interaction networks.
\end{abstract}

\maketitle


\section{Introduction}

Machine learning models for graphs rely on informative features, and spectral features are among the most effective. The spectra of adjacency and Laplacian operators summarize connectivity in a permutation-invariant way and underpin graph kernels and spectral embeddings~\cite{Kriege_2020}. Spectral convolutions built on the graph Laplacian power graph neural networks~\cite{Defferrard_2016, Kipf_2017}, leading modularity eigenvectors enable community detection~\cite{Newman_2006, Umeano_2024_deteqt}, and combinatorial Laplacian spectra encode Betti numbers in topological data analysis~\cite{Carlsson_2009, Scali_2024, Scali_2024_thermal}. A recurring object is the density of states (DOS), a global, basis-invariant spectral summary that quantum protocols can estimate directly without diagonalization~\cite{Scali_2024, Scali_2025_dosqpe}.
\begin{figure*}[t]
    \includegraphics[width=\textwidth]{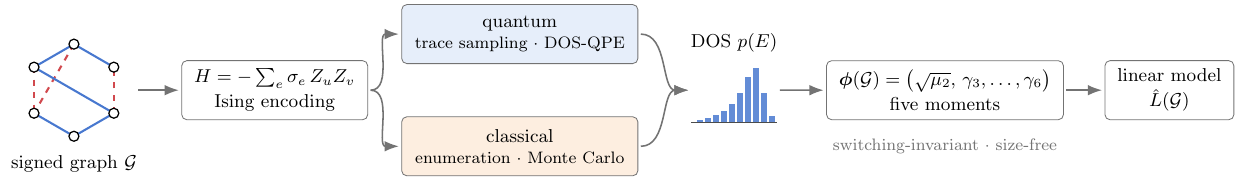}
    \caption{The DOS feature pipeline for signed graphs. A signed graph $\GG$ (positive edges solid, negative edges dashed) is encoded in its native Ising Hamiltonian, Eq.~\eqref{eq:H}. The density of states of $H$ is sampled by either arm of the pipeline: classically, by exact enumeration or Monte Carlo over uniform vertex 2-colorings (Proposition~\ref{prop:sampling}), or on quantum hardware, by trace sampling or by the DOS-QPE circuit of Fig.~\ref{fig:circuit}. Five spectral moments of the sampled density form the feature vector $\bm{\phi}(\GG)$ of Eq.~\eqref{eq:features}, which is switching-invariant and size-free by construction and feeds a classically trained model; we benchmark on the frustration index $L(\GG)$, for which exact labels are available.}
    \label{fig:pipeline}
\end{figure*}

One direction that has received much less attention is that of signed graphs. Here every edge carries a positive or negative label, describing systems built on antagonistic pairwise relations: alliances and rivalries in social networks~\cite{Harary_1953, Read_1954, Sampson_1968}, correlation structures of financial assets~\cite{Sharma_2025}, and quenched couplings of spin glasses~\cite{Anderson_1970, Barahona_1982, Altieri_2024}. A signed graph $\GG$ is \emph{balanced} when every cycle contains an even number of negative edges, and the \emph{frustration index} $L(\GG)$ measures the distance from balance: the minimum number of edge signs one has to flip to reach it. Computing $L$ is NP-hard~\cite{Barahona_1982, Aref_2019}, since it generalizes MAXCUT, which one recovers when all edges are negative. Integer programming solves moderate instances exactly~\cite{Aref_2020}, and signed-graph problems have recently drawn attention from the quantum side~\cite{segawa2021, chen2023quantumcomplexity, Incudini_2024, Kordonowy_2026}. However, spectral feature sets designed for signed graphs remain scarce~\cite{Hameed_2025}, even though the sign structure is precisely where quantities like structural balance reside.

Here, we develop a DOS-based quantum feature extraction and learning approach for signed graphs [Fig.~\ref{fig:pipeline}]. The signed graph is embedded as an Ising Hamiltonian, closely related to Hamiltonian-based encodings for other graph problems~\cite{Hen_2012, Gaitan_2014}. The relevant features are represented by the standardized moments of the Hamiltonian DOS. This construction matches the symmetry of the problem. The moments count signed closed walks through an exact combinatorial identity and are invariant under switching, the gauge transformation that preserves the frustration index. We benchmark the developed approach on the frustration index itself. Although NP-hard in the worst case, it can be labeled exactly at moderate sizes, allowing every prediction to be scored against ground truth. Across a dataset of $1.4\times10^5$ exactly labeled graphs, the exact DOS \emph{determines} $L$, with no two graphs sharing a density of states while differing in frustration index. A regressor using five moments recovers $L$ with a mean absolute error below $0.4$ sign flips at $n=12$ vertices across seventeen classes, while edge and triangle counting incurs twice the error on the same splits. The standardized moments are size-free, allowing the model to transfer across graph sizes and from synthetic ensembles to real-world signed networks.

We highlight that the required features can be extracted natively on quantum hardware, providing access to the DOS without diagonalization. To this end, we simulate two quantum feature extraction routes. The first is a near-term approach that samples the trace of time evolution through Hadamard tests~\cite{Scali_2024}. The second, introduced here as DOS-QPE, targets early fault-tolerant devices and applies quantum phase estimation (QPE) to a purified maximally mixed probe, so that each shot samples directly from the spectral density. Ref.~\cite{Scali_2025_dosqpe} formalizes the corresponding estimation theory and extends the probe beyond maximal mixing. We show that DOS-QPE reaches the noiseless-feature ceiling with orders of magnitude fewer shots than trace sampling and allows its output to be used directly by classically trained models without retraining. At zero field, the pipeline can also be simulated classically. We prove that the DOS is samplable with one edge-list pass per draw (Proposition~\ref{prop:sampling}), allowing feature generation at arbitrary graph size while certifying the quantum routes against exact ground truth. For general dynamics, estimation of the normalized trace of the evolution is complete for the one-clean-qubit class (DQC1)~\cite{Knill_1998, Shor_2008} and is therefore believed to be classically intractable, while the same feature extraction remains available on quantum hardware. We also show that the choice of encoding is essential: the line-graph (spin-ice) alternative is switching-blind, with a spectrum that carries no sign information (Appendix~\ref{app:lemma}).


\section{Spectral features from the density of states}
\label{sec:dos}

The density of states (DOS) of a Hamiltonian $H$ acting on a Hilbert space of dimension $D$, with spectral decomposition $H = \sum_{i=1}^{D} \omega_i |i\rangle\langle i|$, is $S(\omega) = D^{-1}\sum_{i=1}^{D} \delta(\omega - \omega_i)$, normalized here to represent a probability density. Its Fourier transform is the normalized trace of the time evolution,
\begin{equation}
    s(t) = \frac{1}{D}\tr U(t) = \int \! d\omega\, S(\omega)\, e^{-i\omega t},
    \qquad U(t) = e^{-iHt},
    \label{eq:trace_signal}
\end{equation}
so estimating $s(t)$ over time and inverting the transform recovers $S(\omega)$. This forms the basis of the near-term pipeline described in Sec.~\ref{sec:quantum}.

We summarize the DOS through its moments, distinguishing two families. The \emph{energy} moments are the raw spectral averages
\begin{equation}
    M_k = \tr(H^k)/D,
    \label{eq:Mk}
\end{equation}
which carry the combinatorial content. For a Hamiltonian supported on a graph, they expand into counts of closed walks (Sec.~\ref{sec:walks}). For learning features, we instead map the spectrum onto the unit interval by an affine rescaling $\omega \mapsto x$ and use the central moments $\mu_k = \sum_i p_i (x_i - \bar{x})^k$ of the resulting discrete distribution $\{(x_i, p_i)\}$, where $p_i$ is the DOS weight of level $i$ and $\bar{x} = \sum_i p_i x_i$. We define the standardized moments as $\gamma_k = \mu_k/\mu_2^{k/2}$, of which $\gamma_3$ and $\gamma_4$ are the skewness and kurtosis. Given a graph $\GG$ and its Hamiltonian encoding, our feature vector is
\begin{equation}
    \bm{\phi}(\GG) = \left( \sqrt{\mu_2},\, \gamma_3,\, \gamma_4,\, \gamma_5,\, \gamma_6 \right),
    \label{eq:features}
\end{equation}
providing five numbers per graph that capture the width of the rescaled spectrum and its four leading shape parameters. The standardized moments are dimensionless, while rescaling the spectral support makes the full feature set size-free by construction and comparable across graphs of different sizes.

The pipeline, summarized in Fig.~\ref{fig:pipeline}, has four stages: (i) encode the graph in a Hamiltonian whose spectrum carries the structure of interest; (ii) estimate the DOS, classically, by enumeration or, when the encoding permits it, by Monte Carlo (Sec.~\ref{sec:sampling}), or on quantum hardware, by trace sampling or direct spectral sampling; (iii) compress the density into the five moments of Eq.~\eqref{eq:features}; (iv) feed them to a standard, classically trained model. Every stage is problem-agnostic except the first: the encoding decides which invariances the features inherit and which structure they can see. The remainder of the paper instantiates this pipeline for signed graphs, using an Ising Hamiltonian whose spectral features inherit the invariance required by the quantities of interest.


\section{Frustration index}

Next, we recall some core concepts from the theory of signed networks. The frustration index, also known as the line index of balance, is a fundamental measure of structural balance in signed graphs~\cite{Harary_1953}. A signed graph is a graph $\GG = (\VV, \EE, \sigma)$ with $n = |\VV|$ vertices, where $\sigma : \EE \rightarrow \{+1,-1\}$ assigns a sign $\sigma_e$ to each edge $e$, indicating a positive (friendly) or negative (hostile) relationship. For an edge $e = \{u,v\}$, we also write $\sigma_{uv}$. The frustration index measures the ``distance'' of a graph from balance, i.e., from the condition that every cycle contains an even number of negative edges.

Formally, the frustration index $L(\GG)$ of the signed graph $\GG$ is defined as the minimum number of edges whose sign reversal results in a balanced graph,
\begin{equation}
    L(\GG) = \min_{X \subseteq \EE} \left( |X| : \GG^* = (\VV, \EE, \sigma^*_X) \text{ is balanced} \right),
\end{equation}
where $\sigma^*_X$ agrees with $\sigma$ on $\EE \setminus X$ and flips the sign of every edge in $X$. Deleting the edges of $X$ instead of flipping them yields the same minimum~\cite{Aref_2019}, so $L(\GG)$ equally counts the fewest edge deletions that restore balance.

Equivalently, the frustration index can be formulated in terms of vertex 2-colorings as
\begin{equation}
    L(\GG) = \min_{f: \VV \to \{0,1\}} \left| \left\{ \{u,v\} \in \EE : \sigma_{uv} \neq (-1)^{f(u) \oplus f(v)} \right\} \right|,
    \label{eq:coloring}
\end{equation}
where we call an edge \emph{frustrated} by the coloring $f$ when it violates the condition $\sigma_{uv} = (-1)^{f(u) \oplus f(v)}$, and \emph{satisfied} otherwise. This formulation highlights the connection to cut problems: for an all-negative signed graph, Eq.~\eqref{eq:coloring} reduces to MAXCUT, and the decision version of the problem is NP-complete~\cite{Barahona_1982, Aref_2019}. By contrast, deciding balance itself ($L = 0$) is solvable in linear time, illustrating the sharp difference between detecting perfect balance and determining the degree of frustration.
Equation~\eqref{eq:coloring} is also the formulation we use to label our datasets exactly via a standard integer linear program (Appendix~\ref{app:methods})~\cite{Aref_2020, Lubin_2023, Huangfu_2018}.


\section{Ising graph embedding}

We map the frustration index problem onto a spin-glass Ising Hamiltonian with quenched $\pm 1$ couplings supported on the edges of $\GG$~\cite{Anderson_1970, Altieri_2024},
\begin{equation}
    H = \sum_{\{i,j\} \in \EE} J_{ij} Z_i Z_j,
    \label{eq:H}
\end{equation}
where $J_{ij} = -\sigma_{ij}$ and $Z_i$ are Pauli-$Z$ operators. Thus, a positive (friendly) edge gives a ferromagnetic coupling $J_{ij}=-1$, while a negative (hostile) edge gives an antiferromagnetic coupling $J_{ij}=+1$. By construction, $J_{ij}=0$ for non-neighboring vertices.

The frustration index is exactly encoded in the ground-state energy of this Hamiltonian. For a computational-basis state $|s\rangle$, with $s \in \{\pm 1\}^{|\VV|}$ denoting the corresponding $Z$ eigenvalues and interpreted as a vertex 2-coloring, each satisfied edge contributes $-1$ and each frustrated edge $+1$. Hence,
\begin{equation}
    E(s) = 2 f_\GG(s) - |\EE|,
    \qquad
    E_0 = 2 L(\GG) - |\EE|,
    \label{eq:E0_frustration}
\end{equation}
where $f_\GG(s)$ counts the edges frustrated by the coloring $s$. Estimating $L(\GG)$ is therefore equivalent to determining the ground-state energy, motivating us to bypass optimization and instead extract information from spectral statistics.

\emph{Switching invariance.} A switching transformation flips the signs of all edges across a cut $(S, \VV\setminus S)$. It preserves the frustration index and partitions signed graphs into switching-equivalence classes, with $\GG$ balanced iff it is switching-equivalent to the all-positive graph~\cite{Zaslavsky_1982}. On the Hamiltonian side, switching by $S$ is implemented by conjugation with the unitary $U_S = \prod_{i \in S} X_i$, with $X_i$ the Pauli-$X$ operator on vertex $i$: since $X_i Z_i X_i = -Z_i$, the conjugation flips exactly the couplings with one endpoint in $S$,
\begin{equation}
    H_{\sigma'} = U_S\, H_{\sigma}\, U_S^\dagger ,
\end{equation}
where $H_\sigma$ denotes the Hamiltonian of Eq.~\eqref{eq:H} built on the sign function $\sigma$, and $\sigma'$ is the switched assignment. The spectrum of $H$, and hence the DOS and all of its moments, is therefore a \emph{switching-class invariant}, matching the invariance of $L(\GG)$. The DOS features thus discard the gauge redundancy associated with the choice of representative within a switching class, while retaining sign information that is absent from the unsigned graph spectrum. In particular, a balanced graph has the same DOS as its unsigned counterpart.

We stress that not every graph embedding preserves the relevant sign structure. A natural alternative encodes the signed graph in its \emph{line graph}, with edges as sites and sign products on adjacent pairs, in the spirit of spin-ice constructions. We prove in Appendix~\ref{app:lemma} that this route is spectrally sign-blind: the signed line graph is switching-equivalent to the line graph of the underlying unsigned graph, so its spectrum carries no information about frustration. By contrast, the direct encoding of Eq.~\eqref{eq:H} retains the switching-class information on which balance quantities depend, and is therefore the encoding we adopt.


\section{Moments of the Ising density of states}
\label{sec:isingdos}

We first consider the Hamiltonian of Eq.~\eqref{eq:H} without external fields, referring to this case as zero field. Its spectrum is supported on integers $E \in \{-|\EE|, \dots, |\EE|\}$ of fixed parity; a transverse field is introduced later in Sec.~\ref{sec:sampling}. The affine rescaling of Sec.~\ref{sec:dos} then becomes $x = (E + |\EE|)/(2|\EE|)$. Since $\tr H = 0$, the rescaled spectrum has mean $\bar{x} = 1/2$ and carries no graph-dependent information, and the two moment families are related by $\mu_k = M_k/(2|\EE|)^k$ for $k \geq 2$.

The first feature contains only the edge-count information. As shown below, $M_2 = |\EE|$, and hence $\sqrt{\mu_2} = 1/(2\sqrt{|\EE|})$ exactly. We retain this feature because the classifier requires the edge count, while noting that the moment feature set therefore already contains $|\EE|$, which is relevant to the feature-set comparisons below. The remaining four standardized moments $\gamma_3, \dots, \gamma_6$ capture the shape of the DOS and form the size-free part of the representation.

\subsection{Why moments encode frustration}
\label{sec:walks}

The energy moments of the DOS admit an exact combinatorial expansion in terms of the sign structure of the graph. Expanding the $k$-th power of Eq.~\eqref{eq:H} in Eq.~\eqref{eq:Mk}, with $D = 2^{|\VV|}$ and $H = -\sum_{\{i,j\}\in\EE} \sigma_{ij} Z_i Z_j$, the uniform average over computational basis states annihilates every term in which any vertex appears an odd number of times. The surviving terms are ordered $k$-tuples of edges whose multiset has even degree at every vertex and can therefore be decomposed into closed circuits,
\begin{equation}
    M_k = (-1)^k \!\!\sum_{\substack{(e_1,\dots,e_k) \,\in\, \EE^k \\ \text{even vertex degrees}}} \prod_{a=1}^{k} \sigma_{e_a}.
    \label{eq:moment_walks}
\end{equation}
The sign product $\prod_a \sigma_{e_a}$ of any closed circuit is switching-invariant, consistent with the spectral invariance derived above. The lowest moments are
\begin{align}
    M_1 &= 0, \qquad M_2 = |\EE|, \\
    M_3 &= -6 \left( t_+ - t_- \right) = 12\, t_- - 6\, t,
    \label{eq:third_moment}
\end{align}
where $t_\pm$ is the number of balanced/unbalanced triangles and $t = t_+ + t_-$. The third moment, equivalently the skewness of the DOS, counts triangle imbalance directly: unbalanced triangles shift the DOS towards positive skew [Fig.~\ref{fig:dos}], connecting our lowest odd feature to triangle-based balance measures~\cite{Kordonowy_2026}. Even moments are dominated by sign-\emph{independent} pairings of edges, such as the $O(|\EE|^2)$ contribution to $M_4$, with sign-sensitive corrections entering through short signed circuits (squares in $M_4$, pentagons in $M_5$, and so on). Odd moments, by contrast, are purely sign-sensitive. This explains the empirical feature hierarchy observed below: standardized odd moments (skewness, fifth moment) carry the frustration signal, while even moments encode the cycle structure against which it must be normalized. Higher moments probe progressively longer signed circuits, so that the collection $\{M_k\}$ interpolates between local triangle counts and the global cycle-space information that ultimately determines $L(\GG)$ through Eq.~\eqref{eq:E0_frustration}.

\subsection{Classical samplability at zero field}
\label{sec:sampling}

For fixed $k$, Eq.~\eqref{eq:moment_walks} can be evaluated classically in time polynomial in $|\EE|$, so the low-order moments used by our classifier are themselves efficiently computable. At zero field, however, a stronger result holds: the full DOS can be sampled directly.
\begin{figure}[t]
    \includegraphics[width=\columnwidth]{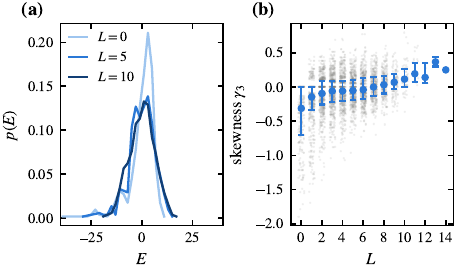}
    \caption{(a)~Exact DOS of three signed graphs with $n = 10$ vertices and $|\EE| = 39$ edges, with frustration index $L = 0, 5, 10$. Frustration drives the spectrum toward symmetry: the skewness rises from $\gamma_3 = -1.90$ to $-0.64$ to $-0.15$ and the kurtosis falls from $\gamma_4 = 8.42$ to $4.16$ to $2.74$, moving towards the Gaussian value of $3$. The lower edge sits at $E_0 = -39, -29, -19$, which is $2L - |\EE|$ in each case, as Eq.~\eqref{eq:E0_frustration} requires. (b)~DOS skewness $\gamma_3$ against the frustration index across the $n = 10$ ensemble ($2 \times 10^4$ graphs; gray: individual graphs, markers: median and interquartile range). The trend reflects Eq.~\eqref{eq:third_moment}, while the spread shows why a single moment is insufficient, motivating the combined features of Eq.~\eqref{eq:features}.}
    \label{fig:dos}
\end{figure}

\begin{proposition}[Monte-Carlo samplability]
\label{prop:sampling}
Let $\GG$ be a signed graph with $m = |\EE|$ edges and let $p(E)$ be the normalized DOS of the Hamiltonian in Eq.~\eqref{eq:H}. Then $p(E)$ is the distribution of the random variable $E(s) = -\sum_{\{u,v\} \in \EE} \sigma_{uv} s_u s_v$ under $s$ drawn uniformly from $\{\pm 1\}^{|\VV|}$. Hence, independent samples from the DOS can be generated classically at $O(m)$ cost per draw, for any graph size.
\end{proposition}

\begin{proof}
The Hamiltonian $H$ is diagonal in the computational basis with diagonal entries $E(s)$, and the DOS assigns every basis state the same weight $2^{-|\VV|}$, which is the uniform distribution over colorings.
\end{proof}

Proposition~\ref{prop:sampling} enables a quantum-inspired version of the zero-field pipeline. It scales feature generation to arbitrary graph size, since classical Monte Carlo reproduces the entire zero-field spectral density, rather than only its low moments, at a measured cost near $1\,$ns per edge per draw. It also gives the quantum pipelines of Sec.~\ref{sec:quantum} an exactly certifiable target, a possibility rarely available on classically hard problems. A similar expansion shows where classical computability ends. A transverse field, $H(h_x) = H + h_x \sum_{i \in \VV} X_i$, keeps every fixed-order moment a closed-form polynomial in $h_x$, since expanding $\tr H(h_x)^k$ in Pauli words annihilates every word carrying an odd number of $X$ factors on any site, with $M_3$ in particular exactly independent of $h_x$. However, the field breaks the diagonal structure on which Proposition~\ref{prop:sampling} rests, so that no classical sampler is known for the density itself or for spectral functionals beyond fixed-order moments, while the hardware routes of Sec.~\ref{sec:quantum} remain available.


\section{Learning the frustration index}
\label{sec:classical}

\emph{Datasets.} For each $n = 6, \dots, 12$, we generate a pool of $2 \times 10^4$ signed graphs from a density-swept Erd\H{o}s--R\'enyi $G(n,p)$ ensemble with $p \in [0.25, 0.75]$. We use two sign generators: independent random signs, and a near-balanced generator, constructed by randomly flipping signs in a balanced graph followed by a random switching, to populate the otherwise rare low- and high-$L$ classes. An integer linear program labels every graph exactly [Eq.~\eqref{eq:coloring}], in milliseconds at $n = 12$. The binding constraint on size is therefore not the labeling but the $2^{|\VV|}$ enumeration behind the exact DOS, which is the constraint that Proposition~\ref{prop:sampling} removes.

\emph{Deduplication.} Graphs sharing the exact energy histogram are indistinguishable to any method based on the DOS. These include isomorphic and switching-equivalent graphs, the latter produced deliberately by our near-balanced generator, as well as accidental cospectral pairs. Splitting such graphs across training and test sets would introduce leakage, so we group the pool by exact histogram and retain one representative per group. The reduction is substantial at small sizes but negligible at large ones: the $2\times10^4$ graphs reduce to $654$ distinct spectra at $n=6$, $11\,402$ at $n=8$, and $19\,862$ at $n=12$. All numbers below are computed on the deduplicated sets, as mean and standard deviation over five stratified 70/30 splits. We cap each class at $1500$ representatives and drop classes with fewer than $50$, leaving $4$ classes at $n = 6$ and $17$ classes ($L = 0, \dots, 16$) at $n = 12$. Further details are given in Appendix~\ref{app:methods}.

\emph{An empirical ceiling.} Grouping by exact histogram also measures how much the representation can possibly deliver. Across all $1.4\times10^5$ labeled graphs, \emph{every} group of graphs sharing a DOS also shares a single value of $L$: no counterexample occurs at any size. On this corpus the exact density of states therefore \emph{determines} the frustration index, and a Bayes-optimal predictor given the exact DOS would recover $L$ without error. The errors reported below therefore arise from compressing the density into five features and from the learner, rather than from the full DOS representation.

\emph{Models and baselines.} Using the features of Eq.~\eqref{eq:features}, we train a multinomial logistic regression that treats the frustration index as a \emph{class}, together with linear and random-forest regressors that treat it as an \emph{ordinal} quantity~\cite{Pedregosa_2011}. Regression performance is measured by the mean absolute error (MAE) of the rounded prediction. We compare against two groups of baselines. The first consists of \emph{counting} features available directly from the edge list: edge count, negative-edge count, triangle count, and signed-triangle sum, which underlie degree- and triangle-based balance indices. The second consists of \emph{spectral balance} measures from the signed-graph literature: $\beta_A = \sum_i (\lambda_i(\Sigma) - \lambda_i(G))^2$, $\gamma_A = \lambda_1(G) - \lambda_1(\Sigma)$, and the smallest signed-Laplacian eigenvalue~\cite{Hameed_2025}. Here, $\Sigma$ and $G$ are the adjacency matrices of the signed graph and its underlying unsigned graph, respectively, with eigenvalues ordered as $\lambda_1 \geq \dots \geq \lambda_n$.
\begin{figure}[t]
    \includegraphics[width=\columnwidth]{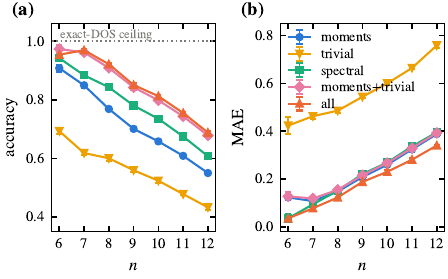}
    \caption{Estimation quality against graph size for the feature sets of Sec.~\ref{sec:classical}, on the deduplicated evaluation sets. Markers and error bars are the mean and standard deviation over five stratified 70/30 splits. (a) Multiclass accuracy of the multinomial logistic regression; the number of classes grows from $4$ at $n=6$ to $17$ at $n=12$. The dotted line is the exact-DOS ceiling: every group of graphs sharing a density of states also shares a single $L$, so a predictor with access to the full density would be exactly right, and the gap below it is the cost of compressing that density into five numbers. (b) Mean absolute error (MAE) of the rounded random-forest regression. Exact-class accuracy necessarily degrades as the label range grows, while the ordinal error stays below $0.4$ sign flips throughout.}
    \label{fig:classical}
\end{figure}

\emph{Results.} Estimation quality against graph size is shown in Fig.~\ref{fig:classical}. At $n = 6$, the five DOS moments alone drive a \emph{linear} classifier to $0.908 \pm 0.011$ accuracy, while adding the counting features reaches $0.973 \pm 0.012$. As $n$ grows, the number of classes more than quadruples, and the moments-only accuracy falls to $0.550 \pm 0.006$ at $n = 12$, compared with $0.433 \pm 0.008$ for the counting features on the same splits. The ordinal view remains sharper: the rounded regression MAE is $0.388 \pm 0.007$ sign flips at $n = 12$ with moments alone and $0.340 \pm 0.003$ with all features, while the counting features give $0.759 \pm 0.010$.

The DOS moments outperform the counting features at every size, by roughly $0.12$ in accuracy and a factor of two in MAE at $n = 12$, showing that the DOS captures signed circuits missed by edge and triangle counts. Classical spectral balance summaries~\cite{Hameed_2025} outperform the DOS moments at most sizes ($0.613 \pm 0.003$ vs $0.581 \pm 0.006$ for the random forest at $n = 12$), showing that frustration leaves a signature in the signed spectrum under different spectral summaries, without privileging our particular choice. The two feature families are complementary, and their union (``all'') gives the best performance throughout, reaching $0.688 \pm 0.007$ accuracy at $n = 12$. This remains below the exact-DOS ceiling of unit accuracy, indicating that a more informative summary of the same density is possible.

\emph{Scope: estimation, not optimization.} At the sizes studied here, exact answers are cheap. Twenty restarts of a textbook greedy 1-flip descent over vertex 2-colorings recover the \emph{exact} frustration index on at least $99.8\%$ of the graphs and always provide a certified upper bound. The search cost is comparable to the $2^{|\VV|}$ enumeration required for our features at $n = 6$ and more than an order of magnitude lower at $n = 12$ (Table~\ref{tab:heuristic}). The purpose of the spectral representation is therefore not to compete with direct optimization at small sizes. Its value lies in providing five switching-invariant, size-free features that any learning model can consume and that can be extracted on quantum hardware.
\begin{table}[b]
\begin{ruledtabular}
\begin{tabular}{lccccc}
$n$ & mean $|L_h - L|$ & exact & $t$(features) & $t$(search) & MAE \\
\hline
 6 & $0.000$ & $100.0\%$ & $4\,\mu$s   & $6\,\mu$s  & $0.123$ \\
 8 & $0.002$ & $99.8\%$  & $27\,\mu$s  & $11\,\mu$s & $0.147$ \\
10 & $0.000$ & $100.0\%$ & $55\,\mu$s  & $10\,\mu$s & $0.260$ \\
12 & $0.002$ & $99.8\%$  & $371\,\mu$s & $18\,\mu$s & $0.388$ \\
\end{tabular}
\end{ruledtabular}
\caption{Multi-restart local search against the learned estimator, on the deduplicated evaluation sets ($500$ graphs per size, median times, single thread). $L_h$ is the frustration index returned by the search; $t$(features) is the exact-DOS enumeration the estimator requires; $t$(search) is the entire heuristic. The last column repeats the estimator's moments-only rounded-regression mean absolute error (MAE) from Fig.~\ref{fig:classical}(b).}
\label{tab:heuristic}
\end{table}

The empirical ceiling stated above makes this precise. The exact density of states determines $L$ on every one of the $1.4 \times 10^5$ labeled graphs, and five moments of that density recover it far above the reach of edge and triangle counting. That is a statement about what the spectrum of Eq.~\eqref{eq:H} encodes, rather than how efficiently it can be computed. A quantity defined through an optimization over $2^{|\VV|}$ colorings leaves a strong signature in a handful of low-order spectral statistics, just as the third moment directly counts signed triangles. This makes the frustration index amenable to estimation from a measured density of states, with the quantum device performing Hamiltonian simulation rather than combinatorial optimization.

\emph{Size transfer.} The standardized moments of Eq.~\eqref{eq:features} are size-free (Sec.~\ref{sec:dos}), allowing models trained at small $n$ to be applied directly at larger sizes. Training on $n \leq 11$ and testing on $n = 12$, the moments-plus-counts model retains $0.566 \pm 0.007$ accuracy and the full feature set $0.621 \pm 0.005$, whereas the spectral-balance summaries, which are not size-normalized, fall to $0.215 \pm 0.001$. Within this range, the combined feature sets transfer best. Moments alone are less accurate than the counting features ($0.350 \pm 0.001$ vs $0.410 \pm 0.006$), while still outperforming them in ordinal error ($0.588 \pm 0.001$ vs $0.773 \pm 0.002$ MAE). The behavior changes sharply when the test graphs lie entirely outside the training-size range.
\begin{figure}[t]
    \includegraphics[width=\columnwidth]{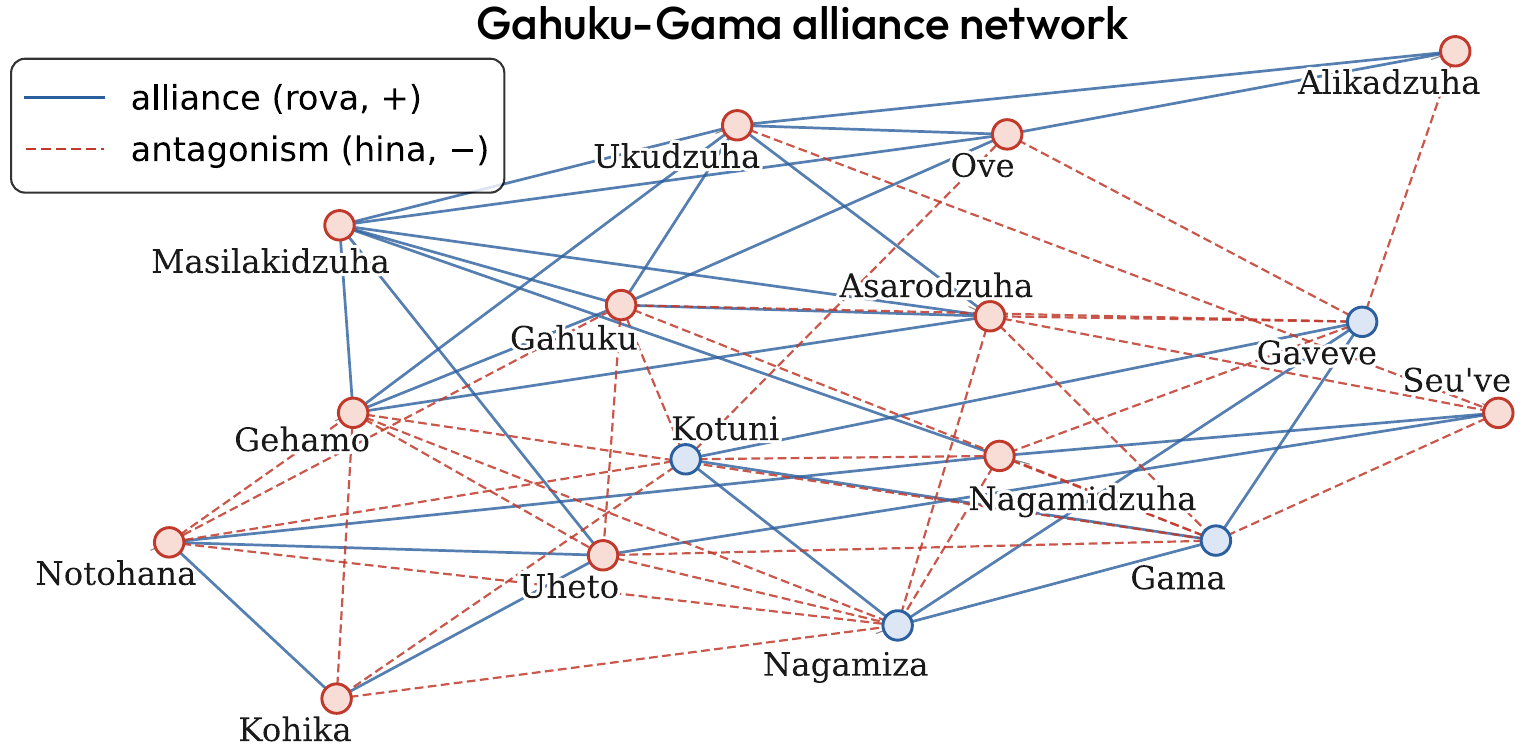}
    \caption{The Gahuku-Gama alliance network~\cite{Read_1954} visualized as a signed graph of $n = 16$ tribes with $29$ positive (\emph{rova}, solid) and $29$ negative (\emph{hina}, dashed) edges corresponding to alliances and antagonistic relations, respectively. Node color marks the two alliance blocs, and the frustration index is $L(\GG) = 7$.}
    \label{fig:network}
\end{figure}

\emph{Real networks.} As an out-of-distribution test we apply models trained \emph{only} on the deduplicated synthetic pools ($n \leq 12$) to two classic signed social networks~\cite{Kunegis_2013}. The first one is the Gahuku-Gama alliance network of the New Guinea highlands ($n = 16$, $|\EE| = 58$, $29$ negative)~\cite{Read_1954}, shown in Fig.~\ref{fig:network}. The second one is the Sampson monastery network ($n = 18$)~\cite{Sampson_1968}, which we discuss below. For Gahuku-Gama the exact frustration index, recomputed with our ILP, is $L = 7$, in agreement with Ref.~\cite{Aref_2019}. Reading only the five moments, and trained on graphs four vertices smaller, the multinomial logistic regression puts $0.574 \pm 0.016$ of its probability on the true $L = 7$, and the linear ordinal regressor predicts $7.09$.

The feature families behave differently under this extrapolation. Adding the raw counting features, despite improving estimation throughout the training-size range, causes the same classifier to assign $0.978$ probability to $L = 9$ and essentially none to the correct class. Edge, negative-edge, and triangle counts grow directly with graph size, so a model relying on them encounters feature values outside its training range and becomes confidently incorrect. By contrast, the normalized DOS features retain their interpretation across graph sizes and transfer successfully even though they are not the most accurate representation in distribution.

The Sampson network~\cite{Sampson_1968} probes the opposite boundary and delimits the training envelope. Symmetrized by the sign of the net pair rating it has $|\EE| = 110$ and $L = 29$, far outside the training label range $L \leq 16$, forcing the estimator to extrapolate. The classifier saturates at its highest available class, while the random-forest regressor returns a frustration density $L/|\EE|$ of $0.119$ against the true $0.264$, with no indication that the input lies outside the training distribution. The practical lesson concerns training pool design: the synthetic ensemble must cover the density and frustration regime of the target graphs. The size-normalized features make this easier to achieve across different graph sizes, since appropriate training data can be generated at any convenient $n$.

Sampson's ratings are directed, so the frustration index itself depends on how the directed ratings are converted into a signed undirected graph. We therefore examine several conventions rather than assigning a unique value. Symmetrizing by the sign of the net pair rating gives $L = 29$ over $110$ edges; keeping zero-net pairs as positive gives $L = 36$ over $126$ edges; calling an edge negative whenever either direction is negative gives $L = 35$; and restricting to the $48$ reciprocated pairs gives $L = 8$. The corresponding frustration densities are $0.264$, $0.286$, $0.278$, and $0.167$. The normalized quantity is therefore more comparable across conventions than $L$ itself, motivating frustration density as a natural target for real-world networks.


\section{Quantum pipelines}
\label{sec:quantum}

The pipeline of Sec.~\ref{sec:classical} used exact spectra. Proposition~\ref{prop:sampling} guarantees that at zero field these spectra are also cheap classically. This makes the frustration index a rare benchmark for quantum protocols: the hardware pipeline can be certified end to end against exact ground truth. Here we perform that certification and ask which quantum route returns the DOS most cheaply, and how much sampling it takes before a downstream classifier stops noticing the difference. We simulate both routes end to end and sample from the exact outcome statistics of each protocol, so that shot noise, the dominant error source, enters exactly.

\emph{NISQ route: trace sampling.} On noisy intermediate-scale quantum (NISQ) hardware, the normalized trace $s(t) = \tr U(t)/2^{|\VV|}$ is estimated by Hadamard tests on a maximally mixed input, with $\text{Re}\, s$ and $\text{Im}\, s$ each obtained from a finite number of binary shots~\cite{Scali_2024}. Because the zero-field spectrum is integer, sampling on the uniform time grid $t_k = 2\pi k / K$ with $K = 2|\EE| + 2$ points suffices for the inverse discrete Fourier transform to recover the DOS exactly in the infinite-shot limit; at finite shots the reconstructed DOS is clipped at zero and renormalized before computing Eq.~\eqref{eq:features}.

\emph{Early fault-tolerant route: DOS-QPE.} We introduce a circuit that samples the DOS directly, which we name DOS-QPE and draw in Fig.~\ref{fig:circuit}: textbook quantum phase estimation run not on an eigenstate but on the maximally mixed state $\rho = \mathds{1}/2^{|\VV|}$, prepared unitarily as half of a maximally entangled state (one Bell pair per system qubit, the purifying register left idle).

\begin{figure}[t]
\centering
\resizebox{\columnwidth}{!}{%
\begin{quantikz}[row sep=0.13cm, column sep=0.20cm]
\lstick[wires=4]{$\ket{0}^{\otimes M}$}
  & \gate{H} & \qw & \ctrl{4} & \qw & \qw & \qw
  & \gate[wires=4]{\text{QFT}^{\dagger}} & \meter{} & \rstick[wires=4]{$y$} \\
  & \gate{H} & \qw & \qw & \ctrl{3} & \qw & \qw & & \meter{} & \\
  & \vdots & & & & & & & \vdots & \\
  & \gate{H} & \qw & \qw & \qw & \qw & \ctrl{1} & & \meter{} & \\
\lstick{$\ket{0}^{\otimes|\VV|}$} & \gate{H^{\otimes|\VV|}} & \ctrl{1}
  & \gate{U^{2^{0}}} & \gate{U^{2^{1}}} & \cdots & \gate{U^{2^{M-1}}} & \qw & \qw & \\
\lstick{$\ket{0}^{\otimes|\VV|}$} & \qw & \targ{} & \qw & \qw & \qw & \qw & \qw & \qw &
\end{quantikz}}
\caption{The DOS-QPE circuit. The lower two registers hold $|\VV|$ qubits each. A layer of Hadamards followed by a transversal \textsc{cnot} prepares one Bell pair per system qubit, so that discarding the purifying register, which takes no further gates and is never measured, leaves the system in $\rho = \mathds{1}/2^{|\VV|}$. Phase estimation then proceeds as usual: the $M$-qubit register controls the powers $U^{2^k}$ of $U = e^{-iH\tau}$ with $\tau = 2\pi/(2|\EE|+2)$, and an inverse quantum Fourier transform is followed by measurement. Because the probe weighs every eigenstate equally, each outcome $y$ is one flat-weight draw from the spectral density, and the empirical moments of the sampled phases are the features of Eq.~\eqref{eq:features}. The register size $M$ sets the resolution of the sampled density and, through Table~\ref{tab:qpebias}, the highest moment order that survives the kernel.}
\label{fig:circuit}
\end{figure}
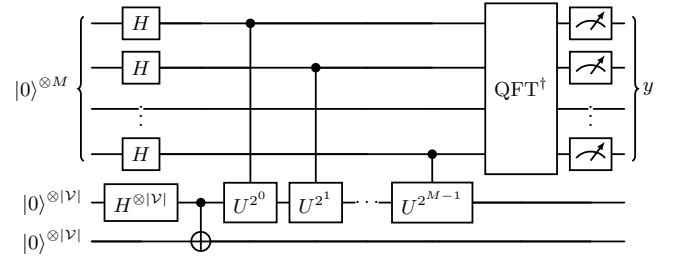 An $M$-qubit phase register applies the controlled powers $U^{2^k}$ of $U = e^{-i H \tau}$, with $\tau = 2\pi/(2|\EE|+2)$, and is read out after an inverse quantum Fourier transform. This $\tau$ maps the integer spectrum onto the distinct eigenphases $\varphi_j = (E_j + |\EE| + 1)/(2|\EE| + 2)$, after the trivial relabeling of register outcomes that absorbs the sign and offset of the raw phase $-E_j \tau / (2\pi) \bmod 1$. The choice essentially maximizes $\tau$ without aliasing: the spectrum spans $2|\EE| + 1$ integer values, so any period of $2|\EE|$ or less would wrap the band edges onto each other; a period of $2|\EE| + 2$, matching the time grid of the trace route, leaves every eigenvalue with its own phase and one empty slot, and any finer $\tau$ would waste register resolution on phases the spectrum never visits. Because the probe weighs all eigenstates equally, the register outcomes $y \in \{0, \dots, 2^M - 1\}$ are distributed according to the spectrum convolved with the $M$-bit QPE kernel, $P(y) = \sum_j 2^{-|\VV|} \, F_M\!\left(y;\, \varphi_j\right)$, with $F_M(y;\varphi) = |\!\sin(\pi 2^M \Delta)/(2^M \sin \pi\Delta)|^2$ and $\Delta = \varphi - y/2^M$. Each shot is one flat-weight draw from the DOS, and the features of Eq.~\eqref{eq:features} are the empirical moments of the sampled phases. Ref.~\cite{Scali_2025_dosqpe} develops the extension beyond maximally mixed probes, where arbitrary purification unitaries select general spectral weights, together with a formalized post-processing and estimation pipeline.
\begin{figure}[t]
    \includegraphics[width=\columnwidth]{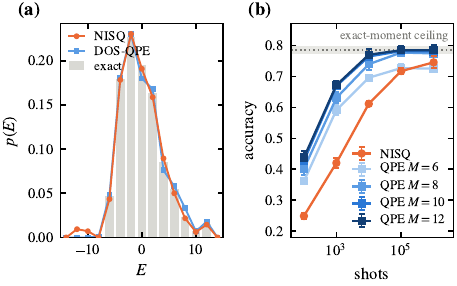}
    \caption{(a)~DOS reconstruction for an $n = 8$, $|\EE| = 14$, $L = 4$ graph at $10^3$ shots, aggregated on the physical energy grid: exact spectrum (gray), NISQ trace sampling (orange), DOS-QPE with $M = 8$ register bits (blue). The NISQ reconstruction develops spurious tail weight, from clipping the noisy inverse Fourier transform, which biases the high-order moments we use as features. (b)~Classifier accuracy at $n = 8$ against shot budget, for models trained and tested on quantum-sampled features. Markers and error bars are the mean and standard deviation over five stratified splits of the deduplicated set; the band is the exact-moment ceiling with its own spread. Two controls act independently: the shot budget sets how fast each curve rises, while the register size $M$ sets the plateau it rises to. A $6$-bit register saturates at $0.727 \pm 0.012$, short of the $0.786 \pm 0.012$ ceiling, whereas $M \geq 10$ reaches the ceiling and stays there. NISQ trace sampling has not caught up even at $10^6$ shots.}
    \label{fig:quantum}
\end{figure}

\emph{Results.} Both routes are run on the same deduplicated graphs the classical estimator was evaluated on, over a full grid of shot budgets and register sizes, with five stratified splits per configuration. Figure~\ref{fig:quantum}(b) shows the $n = 8$ comparison. The two parameters play distinct roles: the shot budget governs how quickly each curve rises by controlling the sampling variance, while the register size determines where it saturates by fixing the kernel bias that cannot be removed by increasing the number of shots.

\emph{(i) Shot efficiency.} DOS-QPE with an $M = 10$ register reaches $0.764 \pm 0.016$ at $10^4$ shots. NISQ trace sampling remains at $0.612 \pm 0.008$ at the same budget, reaches only $0.745 \pm 0.019$ at $10^6$ shots, and does not attain the QPE performance within the range studied. Matching the DOS-QPE result therefore requires more than two orders of magnitude more measurements for the near-term route.

\emph{(ii) The register sets the plateau.} At $n = 8$, the exact-moment ceiling is $0.786 \pm 0.012$, obtained by feeding the same classifier noiseless moments. This is the best performance attainable by any sampler of Eq.~\eqref{eq:features} and is distinct from the exact-DOS ceiling of Sec.~\ref{sec:classical}, which no five-number compression reaches. A $6$-bit register saturates at $0.727 \pm 0.012$ and stops improving beyond $10^5$ shots, leaving an irreducible deficit of about $0.06$. Increasing the register to $M = 8$ raises the plateau to $0.777 \pm 0.011$, while $M = 10$ and $M = 12$ reach $0.782 \pm 0.011$ and $0.785 \pm 0.012$, respectively, indistinguishable from the ceiling. The same pattern holds at $n = 12$, where $M = 10$ and $M = 12$ both converge to $0.559$, against a ceiling of $0.559 \pm 0.008$. This mirrors Table~\ref{tab:qpebias}: the plateau deficit is set by the moment bias imposed by the kernel, so improving beyond it requires more register qubits rather than more shots.

\emph{(iii) Train classically, deploy quantumly.} Training on exact moments and testing on quantum-sampled features represents the realistic deployment mode, since training data come from classical simulation. This transfer succeeds only when the kernel bias is small enough that the classical and quantum feature distributions remain sufficiently aligned. At $n = 8$ transfer accuracy tracks the self-trained value almost exactly for $M = 12$ ($0.789 \pm 0.014$ against $0.784 \pm 0.016$) and for $M = 10$ ($0.768 \pm 0.015$ against $0.781 \pm 0.012$), but falls away at $M = 8$ ($0.597 \pm 0.012$ against $0.774 \pm 0.009$) and collapses at $M = 6$ ($0.422 \pm 0.016$ against $0.725 \pm 0.011$). The NISQ route reaches $0.638 \pm 0.014$ only at $10^6$ shots, because clipping the noisy inverse transform distorts the feature distribution rather than merely broadening it. A classically trained model can therefore be deployed on quantum hardware using sampled features, but only above a register threshold, which for the five-moment feature set is $M \gtrsim 10$.

\emph{Register resolution sets the usable moment order.} The register outcomes follow the DOS convolved with the QPE kernel. As this kernel has $1/\Delta^2$ tails, it places weight far from each eigenvalue. High-order standardized moments are particularly sensitive to this leakage because they increasingly weight the tails. The resulting error is a bias rather than a variance: it does not vanish with more shots, but only with increasing register size. Table~\ref{tab:qpebias} compares both samplers against the exact moments on $300$ graphs at $n = 10$, each using $10^5$ shots.
\begin{table}[b]
\begin{ruledtabular}
\begin{tabular}{lcccc}
moment order & classical MC & $M=8$ & $M=10$ & $M=12$ \\
\hline
 3 & $0.033$ & $0.059$ & $0.034$ & $0.032$ \\
 4 & $0.003$ & $0.066$ & $0.011$ & $0.005$ \\
 6 & $0.007$ & $0.310$ & $0.052$ & $0.017$ \\
12 & $0.022$ & $7.3$   & $1.2$   & $0.24$  \\
20 & $0.034$ & $231$   & $39.8$  & $5.1$   \\
\end{tabular}
\end{ruledtabular}
\caption{Median relative error $|\hat\gamma_k - \gamma_k| / |\gamma_k|$ of standardized moments estimated from $10^5$ draws, against the exact values, for $300$ graphs at $n = 10$. Classical Monte Carlo draws from the true density and its error is pure variance. DOS-QPE draws from the kernel-convolved density, so its error contains a bias that no shot budget removes.}
\label{tab:qpebias}
\end{table}

Classical Monte Carlo, which draws from the true density, maintains small relative errors across all moment orders, with the remaining error due entirely to sampling variance. DOS-QPE matches this accuracy through the fourth moment at $M = 10$ and through the sixth at $M = 12$, but then degrades rapidly with moment order. By order $20$, the estimate is no longer useful for any register considered here. Each additional pair of register qubits reduces the sixth-moment bias by a factor of three to six, while increasing the number of controlled evolutions by a factor of four.

The pipeline of Eq.~\eqref{eq:features} therefore requires $M \gtrsim 10$, set by the highest retained moment rather than by any need to resolve individual spectral levels. This requirement remains mild because the feature set is low-order. Retaining higher moments would require progressively larger registers, which we do not pursue here. At the sizes studied, the additional accuracy from such moments is driven by the lower spectral edge, whose relative weight decreases exponentially with graph size, making this advantage increasingly difficult for any sampling-based approach to retain. Since cumulants add under convolution and the QPE kernel is known in closed form, its contribution could in principle be subtracted rather than resolved, further relaxing the register requirement. We leave this possibility to future work.


\section{Discussion}

Our results are average-case statements about estimation over natural ensembles and two real networks, complementary to the worst-case intractability of exact computation. Typical signed graphs carry substantial spectral information about their frustration index, and the exact DOS determines it across our entire dataset. The five-moment representation remains below this ceiling, while the complementary performance of the spectral-balance summaries of Ref.~\cite{Hameed_2025} suggests that the same density admits more informative compressions. These could include spectral functionals beyond moments, such as tail quantiles or low-lying gaps, which DOS-QPE can access with the same underlying circuit. On the learning side, the real-network tests also establish a requirement on training-pool design: the ensemble must cover the density and frustration regime of the target graphs. The size-normalized features make this easier to satisfy because training can be performed at a convenient graph size.

Interestingly, related approaches to physics-native feature extraction have recently proved useful in other learning settings. Directly measured optical spectra can serve as learning features, with the emission spectrum of a single nonlinear mode enabling classical recognition of squeezed quantum states~\cite{Verstraelen_2026}, while quantum Fourier features extracted from quantum-encoded flow fields enable the detection of vortical structures~\cite{Williams_2025}. For graph learning, photonic positional embeddings generated by light propagation on synthetic frequency lattices can augment graph neural networks~\cite{Wang_2025_PhotonicGCN}, while lattices of polariton condensates can physically encode relational and topological information for subsequent classical learning~\cite{Wang_2026_Polaritonic}. Together with the present results, these examples point to a broader role for physics-native and spectral features as representations that can enhance machine learning.

On the hardware side, the register requirement $M \gtrsim 10$ is set by kernel bias in the higher-order moments and could be relaxed by the kernel-subtraction post-processing outlined in Sec.~\ref{sec:quantum}. The zero-field benchmark is classically samplable (Proposition~\ref{prop:sampling}), which enables exact certification of the quantum extraction routes and provides a controlled boundary between classically accessible and more general dynamics. Beyond zero field, normalized-trace estimation for general dynamics is DQC1-complete~\cite{Knill_1998, Shor_2008}. Fixed-order moments remain classically computable even under a transverse field, but spectral functionals beyond them, beginning with the sampled density itself, have no known efficient classical sampler; DOS-QPE provides access to these while retaining the sample-complexity guarantees of Ref.~\cite{Scali_2025_dosqpe}.

Another motivation is the native realization of Ising models in neutral atom arrays with Rydberg-mediated interactions. These systems provide a promising hardware route for implementing the Ising dynamics required by both pipelines, with programmable Ising-type Hamiltonians available through analog evolution and gate-based control~\cite{Henriet_2020,Gentile_2025}.

On the application side, the same framework extends naturally beyond balance analysis and social networks. Signed interactions arise in spin glasses~\cite{Edwards_1975}, regulatory networks~\cite{Remy_2008}, and protein-interaction networks~\cite{Perfetto_2016}, and related signed-graph problems include correlation clustering~\cite{Bansal_2004}. Across these settings, DOS-based features can provide a compact, switching-invariant summary of global structure that can be extracted on quantum hardware through Hamiltonian simulation.

Finally, the lemma in Appendix~\ref{app:lemma} adds a broader design principle: spectral methods for signed structures are not encoding-agnostic. The line-graph route, natural from a spin ice perspective, provably discards all sign information, whereas the direct encoding of Eq.~\eqref{eq:H} preserves the switching-class information on which balance quantities depend. The effectiveness of the resulting representation therefore comes from matching the Hamiltonian symmetry to the invariances of the learning problem. DOS-based spectral features, already useful for topology and community structure~\cite{Scali_2024, Scali_2024_thermal, Umeano_2024_deteqt}, extend naturally to signed graphs and provide a route from quantum spectral estimation to classical learning without requiring the quantum device itself to solve the underlying combinatorial optimization problem.


\section{Conclusions}

We developed a quantum approach to spectral feature extraction from the density of states of a problem-dependent Hamiltonian and applied it to signed-graph learning. Signed graphs are embedded as Ising Hamiltonians whose energy moments count signed closed walks and yield switching-invariant spectral features. Using the frustration index as a benchmark, we showed that the exact DOS determines the target across $1.4\times10^5$ labeled graphs, while five moments recover it with a mean error of $0.4$ sign flips. We introduced DOS-QPE, which samples the spectral density with orders of magnitude fewer shots than Hadamard-test trace sampling and supports direct use of classically trained models. At zero field, classical DOS sampling enables exact certification. Beyond zero field, the underlying trace-estimation problem is DQC1-complete, providing access to spectral features for which no efficient classical sampler is known.


\begin{acknowledgments}
The authors would like to thank Shaheen Acheche and Louis-Paul Henry from Pasqal for useful discussions. O.\,K. acknowledges the support from UK EPSRC award under the Agreement No. EP/Z53318X/1 (QCi3 Hub).
\end{acknowledgments}


\appendix

\section{Line graphs cannot hear frustration}
\label{app:lemma}

Given a signed graph $\GG = (\VV, \EE, \sigma)$ with underlying unsigned graph $G$, define its \emph{signed line graph} $\Lambda(\GG)$ as the graph whose vertices are the edges $\EE$, with $e \sim f$ whenever $e$ and $f$ share an endpoint, and with edge signs given by the product $\sigma_{ef} = \sigma_e\,\sigma_f$. This is the natural ``spin-ice'' encoding in which frustrated plaquettes of $\GG$ would be represented by interactions among edge variables.

\begin{lemma}
\label{lem:blind}
$\Lambda(\GG)$ is switching-equivalent to $\Lambda(G)$, the all-positive line graph of the underlying graph: with $W = \mathrm{diag}\left(\sigma_e\right)_{e \in \EE}$,
\begin{equation}
    A_{\Lambda(\GG)} = W\, A_{\Lambda(G)}\, W .
    \label{eq:switching_identity}
\end{equation}
Consequently, the adjacency and signed-Laplacian spectra of $\Lambda(\GG)$, and the spectrum of any Ising Hamiltonian of the form of Eq.~\eqref{eq:H} built on $\Lambda(\GG)$, are independent of the sign function $\sigma$. No spectral functional of $\Lambda(\GG)$ can determine $L(\GG)$.
\end{lemma}

\begin{proof}
Entrywise, $[A_{\Lambda(\GG)}]_{ef} = \sigma_e \sigma_f [A_{\Lambda(G)}]_{ef}$ by the definition of the sign product, which is Eq.~\eqref{eq:switching_identity}. Since $W$ is diagonal with entries $\pm 1$, it is orthogonal, so the adjacency spectra coincide. The signed Laplacian (the degree matrix minus the adjacency) conjugates the same way because the degree matrix is diagonal. For the Ising Hamiltonian, conjugation by $U = \prod_{e :\, \sigma_e = -1} X_e$ maps $H_{\Lambda(\GG)}$ to $H_{\Lambda(G)}$, cf.\ the switching argument of the main text. Finally, the frustration index is \emph{not} constant on sign configurations of a fixed underlying graph (the all-positive and all-negative triangles have $L = 0$ and $L = 1$ but identical $\Lambda$ spectra), so no function of the $\Lambda(\GG)$ spectrum can compute it.
\end{proof}

The lemma explains a dead end that, to our knowledge, has not been recorded: any pipeline that extracts spectral statistics (moments, DOS, gaps, spectral measures of balance) from this line-graph encoding returns \emph{exactly} the same answer for every signing of a given graph, and cannot even separate balance from maximal frustration. The information loss occurs at the embedding stage, before any classification method is applied. The statement concerns the product-sign line graph defined above; sign conventions based on bidirected incidences~\cite{Zaslavsky_1982} define different objects and are not covered by (nor needed for) our argument. We verified Eq.~\eqref{eq:switching_identity} as an exact integer identity, along with the coincidence of adjacency and Ising spectra and the variation of $L$ across signings, on $320$ random signed graphs with $n = 5$--$7$ (code accompanies the paper).


\section{Datasets, labeling, and models}
\label{app:methods}

\emph{Exact labeling.} We label every graph by the standard ILP formulation of Eq.~\eqref{eq:coloring}~\cite{Aref_2020}: binary color variables $x_v$, frustration indicators $y_e$, constraints $y_e \geq \pm(x_u - x_v)$ for positive edges and $y_e \geq x_u + x_v - 1$, $y_e \geq 1 - x_u - x_v$ for negative edges, objective $\min \sum_e y_e$, with $x_1 = 0$ breaking the global flip symmetry; solved with HiGHS via JuMP~\cite{Huangfu_2018, Lubin_2023}. Against brute-force enumeration on $80$ graphs ($n = 4$ to $7$) the ILP agrees exactly and runs about $360$ times faster already at $n = 7$. We checked every returned flip set by applying it and confirming the resulting graph has no unbalanced cycle, and every instance in the corpus closed to proven optimality, so no label in this paper rests on a truncated search.

\emph{Pools.} Per size $n \in \{6, \dots, 12\}$: $2\times 10^4$ graphs made unique at generation time by hashing the edge list and sign vector, generators as in Sec.~\ref{sec:classical}, seeds and generator parameters stored with the data. That hash removes literal repeats only; the far stronger deduplication by exact energy histogram described in Sec.~\ref{sec:classical}, which is what the evaluation uses, is applied afterwards and reduces these pools to between $654$ and $19\,862$ distinct spectra. The zero-field spectrum is computed directly as the $2^n$ coloring energies of Eq.~\eqref{eq:E0_frustration}, without constructing any operator.

\emph{Models.} Multinomial logistic regression (no regularization, standardized inputs), random-forest classifier and regressor ($300$ trees), and linear regression~\cite{Pedregosa_2011}. Every reported number is the mean and standard deviation over five stratified 70/30 splits of the deduplicated set, with a per-class cap of $1500$ and a minimum class size of $50$ applied after deduplication. Reported MAEs score rounded regression outputs against the integer label.

\emph{Quantum simulations.} Both pipelines are sampled from exact outcome statistics: binomial shot noise on $\mathrm{Re}\,s(t_k)$, $\mathrm{Im}\,s(t_k)$ for the NISQ route, and multinomial sampling of the kernel-convolved register distribution for DOS-QPE. Graphs are the deduplicated representatives of Sec.~\ref{sec:classical}, stratified by $L$ with a cap of $400$ per class, giving $622$, $2883$ and $6283$ graphs at $n = 6$, $8$ and $12$. The grid is rectangular: every shot budget in $\{10^2, \dots, 10^6\}$ for NISQ, and every pair of $M \in \{6, 8, 10, 12\}$ with the same budgets for DOS-QPE. The kernel-convolved distribution is built once per graph and register size, then sampled at each budget. Simulator moments converge to the exact features in the large-shot limit (deviations ${\sim}3\times 10^{-2}$ in $\gamma_6$ at $10^7$ shots).


\bibliographystyle{unsrtnat} 
\bibliography{main}

\end{document}